\documentclass[a4paper,UKenglish]{lipics}
 
\usepackage{microtype}
\usepackage[utf8]{inputenc}
\usepackage{algorithm2e}
\usepackage{tikz}
\usepackage{graphicx}
\usepackage{hyperref}
\usepackage{xspace}							

\usetikzlibrary{arrows,positioning}

\title{Fixed parameter complexity of distance constrained labeling and uniform channel assignment problems
\footnote{Paper supported by project Kontakt LH12095 and by GAUK project 1784214.}\footnote{Second, third and fourth author are supported by the project SVV--2015--260223. First, third and fifth author is supported by project CE-ITI P202/12/G061 of GA ČR.}}
\titlerunning{Distance constrained labeling} 

\author[1]{Jiří Fiala}
\author[1]{Tomáš Gavenčiak}
\author[1]{Dušan Knop}
\author[1]{Martin Koutecký}
\author[1]{Jan Kratochvíl}
\affil[1]{Department of Applied Mathematics, Charles University\\
  Malostranské nám. 25, Prague
  \texttt{\{fiala,gavento,knop,koutecky,honza\}@kam.mff.cuni.cz}
}
\authorrunning{J. Fiala, T. Gavenčiak, D. Knop, M. Koutecký and J. Kratochvíl} 

\Copyright{Jiří Fiala, Tomáš Gavenčiak, Dušan Knop, Martin Koutecký and Jan Kratochvíl}


\DeclareMathOperator{\dist}{{\rm dist}}
\DeclareMathOperator{\nd}{{\rm nd}}
\DeclareMathOperator{\nlc}{{\rm nlc}}
\DeclareMathOperator{\vc}{{\rm vc}}

\newcommand{\cT}{{\mathcal T}}
\newcommand{\cW}{{\mathcal W}}

\newcommand{\NP}{{\sf{NP}}\xspace}
\newcommand{\FPT}{{\sf{FPT}}\xspace}
\newcommand{\N}{{\mathbb{N}}}

\newcommand{\bw}{{\mathbf w}}

\def\romanitems{
\renewcommand{\theenumi}{\roman{enumi}}
\renewcommand{\labelenumi}{(\theenumi)}
}

\theoremstyle{plain}
\newtheorem{observation}[theorem]{Observation}
\newtheorem{proposition}[theorem]{Proposition}
\theoremstyle{definition}
\newtheorem{prob}[theorem]{Problem}

\serieslogo{}
\volumeinfo
  {Billy Editor and Bill Editors}
  {2}
  {Conference title on which this volume is based on}
  {1}
  {1}
  {1}
\EventShortName{}
\DOI{10.4230/LIPIcs.xxx.yyy.p}

\begin{document}

\maketitle

\begin{abstract}
  We study computational complexity of the class of distance-constrained graph labeling problems
from the fixed parameter tractability point of view. The parameters studied are 
neighborhood diversity and clique width. 


We rephrase the distance constrained graph labeling problem as
a specific uniform variant of the {\sc Channel Assignment} problem and show that this problem is fixed parameter tractable when 
parameterized by the neighborhood diversity together with the largest weight.
Consequently, every {\sc $L(p_1, p_2,\dots, p_k)$-labeling} problem is \FPT{} when parameterized by
the neighborhood diversity, the maximum $p_i$ and $k$.

Our results yield also \FPT{} algorithms for all {\sc $L(p_1, p_2,\dots, p_k)$-labeling} 
problems when parameterized by the size of a minimum vertex cover, answering an open question of
Fiala et al.: \emph{Parameterized complexity of coloring problems: Treewidth versus
vertex cover}. The same consequence applies on {\sc Channel Assignment} when the maximum weight is additionally
included among the parameters.


Finally, we show that the uniform variant of the {\sc Channel Assignment} problem becomes \NP-complete when 
generalized to graphs of bounded clique width.
\end{abstract}

\section{Introduction}\label{s:intro}

The frequency assignment problem in wireless networks 
yields an abundance of various mathematical models and related problems.
We study a group of such discrete optimization problems in terms of parameterized
computational complexity, which is one of the central paradigms of
contemporary theoretical computer science. We study parameterization of the problems by
\emph{clique width} and particularly by \emph{neighborhood diversity} ($\nd$), 
a graph parameter lying between clique width and the size of a minimum vertex cover.

All these problems are \NP-hard even for constant clique width,
including the uniform variant, as we show in this paper. 
On the other hand, we prove that they are in \FPT{} w.r.t. $\nd$. 
Such fixed parameter tractability has been so far known only for the special case of $L(p,1)$
labeling when parameterized by vertex cover~\cite{l:FGK09}.

\subsection{Distance constrained labelings}

Given a $k$-tuple of positive integers $p_1,\dots,p_k$, called \emph{distance constraints}, an $L(p_1,\dots,p_k)$-labeling of a graph is an assignment $l$ of integer labels to the vertices of the graph satisfying the following condition: 
Whenever vertices $u$ and $v$ are at distance $i$, the assigned labels differ by at least $p_i$. 
Formally, $\dist(u,v)=i \Longrightarrow |l(u)-l(v)|\ge p_i$ for all $u,v: \dist(u,v)\le k$. 
Often only non-increasing sequences of distance constraints are considered.

Any $L(1)$-labeling is a graph coloring and vice-versa. Analogously, any coloring of the $k$-th distance power of a graph is an $L(1,\dots,1)$-labeling. 
The concept of $L(2,1)$-labeling is attributed to Roberts by Griggs and Yeh~\cite{l:GY92}. 
It is not difficult to show that whenever $l$ is an optimal $L(p_1,\dots,p_k)$-labeling within a range $[0,\lambda]$, 
then the so called \emph{span} $\lambda$ is a linear combination of $p_1,\dots,p_k$~\cite{l:GY92,l:Kral06}. 
In particular, a graph $G$ allows an $L(p_1,\dots,p_k)$-labeling of span $\lambda$ if and only it has an $L(cp_1,\dots,cp_k)$-labeling of span $c\lambda$ for any positive integer $c$.

For the computational complexity purposes, we define the following class of decision problems:

\begin{prob}[\sc $L(p_1,\dots,p_k)$-labeling]~
\begin{center}
\begin{tabular} {|ll|}
      \hline
      {\bf Parameters:\enspace} & {\parbox[t]{27em}{Positive integers $p_1,\dots,p_k$}}\\
      {\bf Input:\enspace} & {\parbox[t]{27em}{Graph $G$, positive integer $\lambda$}}\\
      {\bf Question:\enspace}&\parbox[t]{27em}{Is there an $L(p_1,\dots,p_k)$ labeling of $G$ using labels from the interval $[0,\lambda]$?}\\
      \hline
\end{tabular}
\end{center}
\end{prob}

The {\sc $L(2,1)$-labeling} problem has been shown to be \NP-complete by Griggs and Yeh~\cite{l:GY92} 
by a reduction from {\sc Hamiltonian cycle} (with $\lambda=|V_G|$). 
Fiala, Kratochvíl and Kloks~\cite{l:FKK01} showed that {\sc $L(2,1)$-labeling} 
remains \NP-complete also for all fixed $\lambda\ge 4$,
while for $\lambda\le 3$ it is solvable in linear time. 

Despite a conjecture that {\sc $L(2,1)$-labeling} remains \NP-complete on trees~\cite{l:GY92}, 
Chang and Kuo~\cite{l:ChK96} showed a dynamic programming algorithm for this problem, 
as well as for all {$L(p_1,p_2)$-labelings where $p_2$ divides $p_1$. 
All the remaining cases have been shown to be \NP-complete by Fiala, Golovach and Kratochvíl~\cite{l:FGK08}. 
For graphs of tree width 2, the same authors show that {\sc $L(2,1)$-labeling} is \NP complete already on series-parallel graphs~\cite{l:FGK05}. 
Note that these results imply \NP-hardness 
of {\sc $L(3,2)$-labeling} on graphs of clique width at most 3 and of {\sc $L(2,1)$-labeling} for clique width at most 9.

On the other hand, when $\lambda$ is fixed, then the existence of an $L(p_1,\dots,p_k)$-labeling of $G$ can be expressed in MSO$_1,$ hence it allows a linear time algorithm on any graph of bounded clique width~\cite{t:KR01}.

Fiala et al.~\cite{l:FGK09} showed that the problem of {\sc $L(p,1)$-labeling} is \FPT{} when
parameterized by $p$ together with the size of the vertex cover. 
They also ask for the complexity characterization of the related {\sc Channel Assignment} problem.
We extend their work to the broader class of graphs and, consequently, in our Theorem~\ref{thm:chaVC} we provide a solution for their open problem.

\subsection{Channel assignment}

Channel assignment is a concept closely related to distance constrained graph labeling. 
Here, every edge has a prescribed weight $w(e)$ and it is required
that the labels of adjacent vertices differ at least by the weight of the corresponding edge. 
The associated decision problem is defined as follows:

\begin{prob}[\sc Channel Assignment]~
\begin{center}
\begin{tabular} {|ll|}
      \hline
      {\bf Input:\enspace} & {\parbox[t]{27em}{Graph $G$, a positive integer $\lambda$, edge weights $w: E_G\to \N$}}\\
      {\bf Question:\enspace}&\parbox[t]{27em}{Is there a labeling $l$ of the vertices of $G$ by integers from $[0,\lambda]$
                                                such that $|l(u)-l(v)|\geq w(u,v)$ for all $uv\in E_G$?}\\
      \hline
\end{tabular}
\end{center}
\end{prob}

The maximal edge weight is an obvious necessary lower bound for the span of any labeling. 
Observe that for any bipartite graph, in particular also for all trees, it is also an upper bound 
--- a labeling that assigns $0$ to one class of the bipartition and $w_{\max}=\max\{w(e),e\in E_G\}$ to the other class
satisfies all edge constraints. 
McDiarmid and Reed \cite{c:DR03} showed that it is \NP-complete to decide whether a graph of tree width $3$ 
allows a channel assignment of given span $\lambda$.
This \NP-hardness hence applies on graphs of clique width at most $17$.
It is worth to note that for graphs of tree width $2$, i.e. for subgraphs of series-parallel graphs,
the complexity characterization of the {\sc Channel Assignment} is still open. Only few partial results are known~\cite{c:Skvarek10},
among others that the {\sc Channel Assignment} is polynomially solvable on graphs of bounded tree width if the span $\lambda$ is bounded by a constant.

Any instance $G$, $\lambda$ of the {\sc $L(p_1,\dots,p_k)$-labeling} problem can straightforwardly be reduced to an instance $G^k,\lambda,w$ of the {\sc Channel Assignment} problem.
Here, $G^k$ arises from $G$ by connecting all pairs of vertices that are in $G$ at distance at most $k$, and for the edges of $G^k$
we let $w(u,v)=p_i$ whenever $\dist_G(u,v)=i$.

The resulting instances of {\sc Channel Assignment} have by the construction some special properties. We explore and generalize these to obtain a uniform variant of the {\sc Channel Assignment} problem.

\subsection{Neighborhood diversity}

Lampis significantly reduced (from the tower function to double exponential) the hidden constants of the generic polynomial algorithms for MSO$_2$ model checking on graphs with bounded vertex cover~\cite{t:Lampis12}.
To extend this approach to a broader class of graphs he introduced a new graph parameter called the neighborhood diversity of a graph as follows:

\begin{definition}[Neighborhood diversity]
A partition $V_1,\dots,V_d$ is called a \emph{neighborhood diversity decomposition} if it satisfies
\begin{itemize}
\item each $V_i$ induces either an empty subgraph or a complete subgraph of $G$, and
\item for each distinct $V_i$ and $V_j$ there are either no edges between $V_i$ and $V_j$, or every vertex of $V_i$ is adjacent to all vertices of $V_j$.
\end{itemize}
We write $u \sim v$ to indicate that $u$ and $v$ belong to the same class of the decomposition.

The \emph{neighborhood diversity} of a graph $G$, denoted by $\nd(G)$, is the minimum number of classes of a neighborhood diversity decomposition.
\end{definition}

Observe that for the optimal neighborhood diversity decomposition it holds that $u\sim u'$ is equivalent with 
$N(u)\setminus v = N(v)\setminus u$. 
Therefore, the optimal neighborhood diversity decomposition can be computed in $O(n^3)$ time~\cite{t:Lampis12}.

Classes of graphs of bounded neighborhood diversity reside between classes of bounded vertex cover and graphs of bounded clique width. Several non-MSO$_1$ problems, e.g. {\sc Hamiltonian cycle} not be solved in polynomial time on graphs of bounded clique width~\cite{t:Wanke94}. On the other hand, Fomin et al. stated more precisely that the {\sc Hamiltonian cycle} problem is $W[1]$-hard, when parameterized by clique width~\cite{t:FGLS10}.
In sequel, Lampis showed that some of these problems, including {\sc Hamiltonian cycle}, are indeed fixed parameter tractable on graphs of bounded neighborhood diversity~\cite{t:Lampis12}.

Ganian and Obdržálek~\cite{t:GO13} further deepened Lampis' results and showed that also  problems expressible in MSO$_1$
with cardinality constraints (cardMSO$_1$) are fixed parameter tractable when parameterized by $\vc(G)$ and/or $\nd(G)$.

It is easy to see that for a graph $G$ it holds that $\nd(G)\le2^{\vc(G)} + \vc(G)$
where $\vc(G)$ is the size of minimal vertex cover of the graph $G$. This is also used in more detail in the proof of 
Theorem~\ref{thm:chaVC}.

Observe that a sufficiently large $n$-vertex graph of bounded neighborhood diversity can be described in significantly more effective way, namely by using only $O(\log n\nd(G)^2)$ space:

\begin{definition}[Type graph]
The \emph{type graph} $T(G)$ for a neighborhood diversity decomposition $V_1,\dots,V_d$
of a graph $G$ is a vertex weighted graph on vertices $\{t_1,\dots,t_{d}\}$, 
where each $t_i$ is assigned weight $s(t_i)=|V_i|$, i.e. the size of the corresponding class of the decomposition. 
Distinct vertices $t_i$ and $t_j$ are adjacent in $T(G)$ if and only if the edges between the two corresponding classes $V_i$ and $V_j$ form a complete bipartite graph. Moreover, $T(G)$ contains a loop incident with vertex $t_i$ if and only if the corresponding class $V_i$ induces a clique. 
\end{definition}

\begin{figure}[ht!]
  \begin{tikzpicture}
  \tikzstyle{vrchol}=[circle, draw, inner sep=1pt, minimum width=3pt]
  \tikzstyle{node}=[fill=yellow!50, yellow!50]
  \tikzstyle{headder}=[font=\normalfont]
  
  \begin{scope}[shift={(0,4)},node distance=0.7cm]
    \node[headder] at (1.5,4) {$G$ and its $L(2,1,1)$-labelling};

    \draw[node] (0,0) circle (1cm);
    \node[vrchol](v1) at (0, 0.35) {7};
    \node[vrchol, below of=v1](v2) {5};

    \draw[node] (3,0) circle (1cm);
    \node[vrchol](v3) at (3,0) {3};

    \draw[node] (6,0) circle (1cm);
    \node[vrchol](v4) at (6,0) {8};

    \draw[node] (9,0) circle (1cm);
    \node[vrchol](v5) at (9, 0.35) {1};
    \node[vrchol, below of=v5](v6) {0};

    \draw[node] (12,0) circle (1cm);
    \node[vrchol](v7) at (12, 0.7) {7};
    \node[vrchol, below of=v7](v8) {9};
    \node[vrchol, below of=v8](v9) {5};

    \draw[node] (4.5,3) circle (1cm);
    \node[vrchol](v10) at (3.8, 3) {2};
    \node[vrchol, right of=v10](v11) {4};
    \node[vrchol, right of=v11](v12) {6};

    \draw (v1) edge (v2);
    \draw (v1) edge (v3);
    \draw (v2) edge (v3);
    \draw (v3) edge (v4);
    \draw (v4) edge (v5);
    \draw (v4) edge (v6);
    \draw (v5) edge (v7);
    \draw (v5) edge (v8);
    \draw (v5) edge (v9);
    \draw (v6) edge (v7);
    \draw (v6) edge (v8);
    \draw (v6) edge (v9);

    \draw (v7) edge (v8);
    \draw (v8) edge (v9);
    \draw (v7) edge[bend left] (v9);
    
    \draw (v3) edge (v10);
    \draw (v3) edge (v11);
    \draw (v3) edge (v12);
    \draw (v4) edge (v10);
    \draw (v4) edge (v11);
    \draw (v4) edge (v12);
    \draw (v10) edge (v11);
    \draw (v11) edge (v12);
    \draw (v10) edge[bend left] (v12);
  \end{scope}

  \begin{scope}[shift={(0,-0.5)},node distance=1cm]
    \node[headder] at (0,2.5) {$T(G)$};

    \node[vrchol](n1) at (0,0) {2};
    \node[vrchol, right of=n1](n2) {1};
    \node[vrchol, right of=n2](n3) {1};
    \node[vrchol, right of=n3](n4) {2};
    \node[vrchol, right of=n4](n5) {3};
    \node[vrchol](n6) at (1.5,1.5) {3};

    \path (n1) edge [out=230, in=310, distance=1cm] node[above] {} (n1);
    \path (n2) edge [out=230, in=310, distance=1cm] node[above] {} (n2);
    \path (n3) edge [out=230, in=310, distance=1cm] node[above] {} (n3);
    \path (n5) edge [out=230, in=310, distance=1cm] node[above] {} (n5);
    \path (n6) edge [out=50, in=130, distance=1cm] node[below] {} (n6);

    \path (n1) edge (n2);
    \path (n2) edge (n3);
    \path (n3) edge (n4);
    \path (n4) edge (n5);
    \path (n2) edge (n6);
    \path (n3) edge (n6);
    
  \end{scope}

  \begin{scope}[node distance=1.5cm,shift={(5.5,-0.5)}]
    \node[headder] at (0.5,2.5) {$T(G^3)$ and $w$};

    \node[vrchol](u1) at (0,0) {2};
    \node[vrchol, right of=u1](u2) {1};
    \node[vrchol, right of=u2](u3) {1};
    \node[vrchol, right of=u3](u4) {2};
    \node[vrchol, right of=u4](u5) {3};
    \node[vrchol](u6) at (2.25,1.5) {3};

    \path (u1) edge [out=230, in=310, distance=1cm] node[above] {2} (u1);
    \path (u2) edge [out=230, in=310, distance=1cm] node[above] {2} (u2);
    \path (u3) edge [out=230, in=310, distance=1cm] node[above] {2} (u3);
    \path (u4) edge [out=230, in=310, distance=1cm] node[above] {1} (u4);
    \path (u5) edge [out=230, in=310, distance=1cm] node[above] {2} (u5);
    \path (u6) edge [out=50, in=130, distance=1cm] node[below] {2} (u6);

    \path (u1) edge node[above] {2} (u2);
    \path (u2) edge node[above] {2} (u3);
    \path (u3) edge node[above] {2} (u4);
    \path (u4) edge node[above] {2} (u5);

    \path (u1) edge node[left] {1} (u6);
    \path (u2) edge node[left] {2} (u6);
    \path (u3) edge node[left] {2} (u6);
    \path (u4) edge node[left] {1} (u6);
    \path (u5) edge node[left] {1} (u6);

    \path (u1) edge[out=340, in=220, distance=1.5cm] node[below right] {1} (u3);
    \path (u1) edge[out=320, in=220, distance=3cm] node[below] {1} (u4);

    \path (u5) edge[out=190, in=320, distance=1.5cm] node[below left] {1} (u3);
    \path (u5) edge[out=220, in=340, distance=3cm] node[below] {1} (u2);
  \end{scope}
\end{tikzpicture}
  \caption{An example of a graph with its neighborhood diversity decomposition. 
	Vertex labels indicate one of its optimal $L(2,1,1)$-labelings. 
  The corresponding type graph. 
	The weighted type graph corresponding to the resulting instance of the {\sc Channel Assignment} problem.}
	\label{fig:labelExample}
\end{figure}
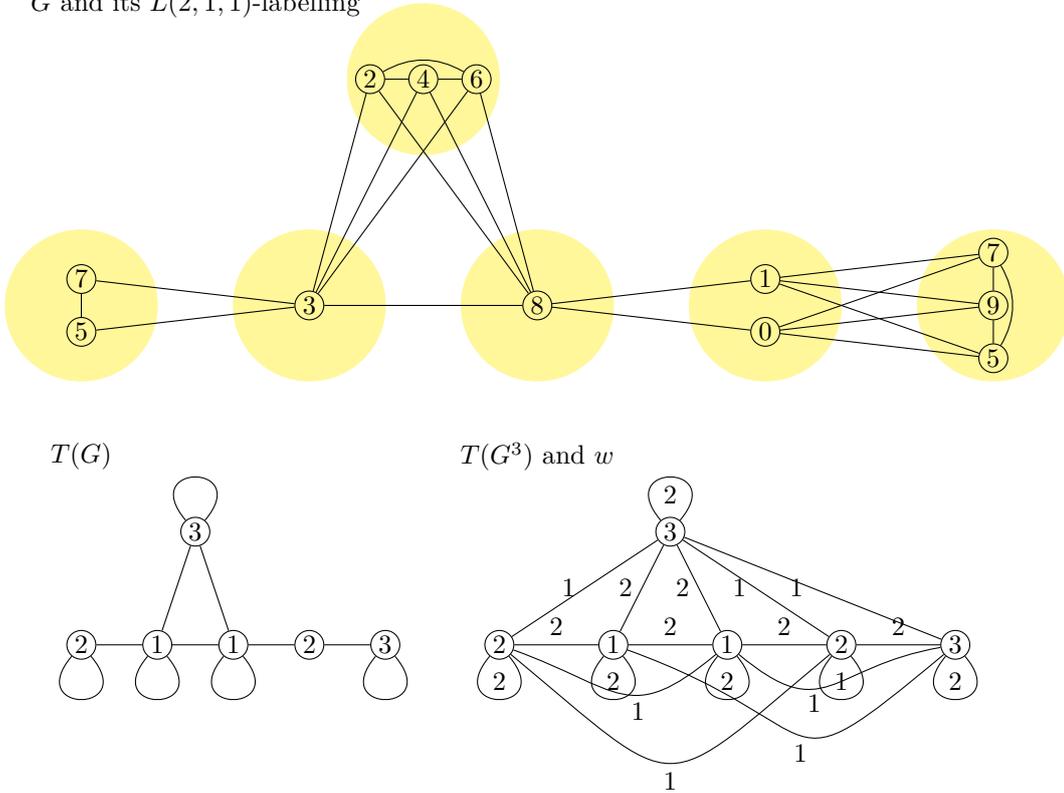

For our purposes, i.e. to decide existence of a suitable labeling of a graph $G$, 
it suffices to consider only its type graph, as $G$ can be uniquely reconstructed from $T(G)$ 
(upto an isomorphism) and vice-versa.

Moreover, the reduction of {\sc $L(p_1,\dots,p_k)$-labeling} to {\sc Channel Assignment} preserves the property of bounded neighborhood diversity:

\begin{observation}
For any graph $G$ and any positive integer $k$ it holds that $\nd(G)\ge\nd(G^k)$.
\end{observation}

\begin{proof}
The optimal neighborhood diversity decomposition of $G$ is a neighborhood diversity decomposition of $G^k$. 
\end{proof}

\subsection{Our contribution}\label{ss:our}

Our goal is an extension of the FPT algorithm for {\sc $L(2,1)$-labeling} on graphs of bounded vertex cover
to broader graph class and for rich collections of distance constraints.
In particular, we aim at {\sc $L(p_1,\dots,p_k)$-labeling} on graphs of bounded neighborhood diversity.

For this purpose we utilize the aforementioned reduction to the {\sc Channel Assignment}, taking into account that the
neighborhood diversity remains bounded, even though the underlying graph changes.

It is worth to note that we must adopt additional assumptions for the {\sc Channel Assignment} since otherwise 
it is \NP-complete already on complete graphs, i.e. on graphs with $\nd(G)=1$.
To see this, we recall the construction of Griggs and Yeh~\cite{l:GY92}. They show that a graph 
$H$ on $n$ vertices has a Hamiltonian path if and only if the complement $H$ extended by a single universal vertex allows 
an $L(2,1)$-labeling of span $n+1$. As the existence of a universal vertex yields diameter two, the underlying graph for the resulting instance of {\sc Channel Assignment} is $K_{n+1}$.

On the other hand, the additional assumptions on the instances of {\sc Channel Assignment} still shall allow to reduce 
any instance of the {\sc $L(p_1,\dots,p_k)$-labeling} problem.
By the reduction, all edges between classes of the neighborhood diversity decomposition are assigned the same weight.
We formally adopt this as our additional constraint as follows:

\begin{definition}
The edge weights $w$ on a graph $G$ are \emph{$\nd$-uniform} if $w(u,v)=w(u', v')$ whenever $u\sim u'$ and $v\sim v'$
w.r.t. the optimal neighborhood diversity decomposition. 
In a similar way we define uniform weights w.r.t. a particular decomposition.
\end{definition}

Our main contribution is an algorithm for the following scenario:

\begin{theorem}\label{thm:cha}
The {\sc Channel Assignment} problem on $\nd$-uniform instances is \FPT{} when parameterized by $\nd$ and $w_{\max}$, where $w_{\max}=\max\{w(e),e\in E_G\}$.
\end{theorem}

Immediately, we get the following consequence:

\begin{theorem}\label{thm:Lp}
The $p_1,\dots,p_k,$ the {\sc $L(p_1,\dots,p_k)$-labeling} problem is \FPT{} when parameterized by $\nd$, $k$ and maximum $p_i$ (or equivalently by $\nd$ and the $k$-tuple $(p_1,\dots,p_k)$).
\end{theorem}

Furthermore, our \FPT{} result for {\sc Channel Assignment} extends to vertex cover even without the uniformity requirement.

\begin{theorem}\label{thm:chaVC}
The {\sc Channel Assignment} problem is \FPT{} when parameterized by $w_{\max}$ and the size of vertex cover.
\end{theorem}

One may ask whether the uniform version of {\sc Channel Assignment} allows an \FPT{} algorithm also for a broader class of graph.
Finally, we show that a natural generalization of this concept on graphs of bounded clique width yields an \NP-complete 
problem on graphs of clique width at most 5.

\section{Representing labelings as sequences and walks}

We now focus on the $\nd$-uniform instances of the {\sc Channel Assignment} problem.
It has been already mentioned that the optimal neighborhood diversity decomposition can be computed in cubic time. 
The test, whether it is $\nd$-uniform, could be computed in extra quadratic time. 
On the other hand, on $\nd$-uniform instances it suffices to consider only the type graph, whose edges take weights from the edges of the underlying graph (see Fig.~\ref{fig:labelExample}), since such weighted type graph corresponds 
uniquely to the original weighted graph, upto an isomorphism.

Hence without loss of generalization 
assume that our algorithms are given the type graph whose edges are weighted by separation constraints $w$,
however we express the time complexity bounds in terms of the size of the original graph.

Without loss of generality we may assume that the given graph $G$ and its type graph $T(G)$ are connected, since connected components can be treated independently. 

If the type graph $T(G)$ contains a type $t$ not incident with a loop, we may reduce the channel assignment 
problem to the graph $G'$, obtained from $G$ by deleting all but one vertices of the type $t$. 
Any channel assignment of $G'$ yields a valid channel assignment of $G$ by using the same label on all vertices of type $t$ in $G$ as was given to the single vertex of type $t$ in $G'$. Observe that adding a loop to a 
type, which represents only a single vertex, does not affect the resulting graph $G'$. Hence we assume without loss of generality that all types are incident with a loop. We call such type graph \emph{reflexive}. 

\begin{observation}\label{obs:refl}
If the type graph $T(G)$ is reflexive, then vertices of $G$ of the same type have distinct labels in every channel assignment.
\end{observation}

Up to an isomorphism of the graph $G$, any channel assignment $l$ is uniquely characterized by a sequence of type sets as follows:

\begin{lemma}\label{lem:seq}
Any weighted graph $G$ corresponding to a reflexive weighted type graph $T(G), w$ 
allows a channel assignment of span $\lambda$,
if and only if there exists a sequence of sets 
$\cT=T_0,\dots,T_{\lambda}$ of the following properties:
\begin{enumerate}\romanitems
 \item $T_i\subseteq V_{T(G)}$ for each $i\in [0,\lambda]$, 
 \item for each $t\in V_{T(G)}: s(t)=|\{T_i: t\in T_i\}|$, 
 \item for all $(t,r)\in E_{T(G)}: ( t\in T_i \land r\in T_j ) \Rightarrow |i-j|\ge w(t,r)$
\end{enumerate}
\end{lemma}

\begin{proof}
Given a channel assignment $l: V_G \to [0,\lambda]$, we define the desired sequence $\cT$, such that the $i$-th element is the set of types that contain a vertex labeled by $i$. Formally $T_i=\{t: \exists u\in V_t : l(u)=i\}$.
Now
\begin{enumerate}\romanitems
\item each element of the sequence is a set of types, possibly empty,
\item as all vertices of $V_i$ are labeled by distinct labels by Observation~\ref{obs:refl}, any type $t$ occurs in $s(t)$ many elements of the sequence
\item if $u$ of type $t$ is labeled by $i$, and it is adjacent to $v$ of type $r$ labeled by $j$, then
$|i-j|=|l(u)-l(v)|\ge w(u,v) = w(t,r)$, i.e. adjacent types $t$ and $r$ may appear in sets 
that are in the sequence at least $w(t,r)$ apart.
\end{enumerate}

In the opposite direction assume that the sequence $\cT$ exists. Then for each set $T_i$ and type $t_j\in T_i$ we choose a
distinct vertex $u\in V_j$ and label it by $i$, i.e. $l(u)=i$.

Now the condition (ii) guarantees that all vertices are labeled, while condition (iii) guarantees that all distance constraints are fulfilled.
\end{proof}

Observe that Lemma~\ref{lem:seq} poses no constrains on sets $T$ that are at distance at least $w_{\max}$. Hence, we build an auxiliary directed graph $D$ on all possible sequences of sets of length at most $z=w_{\max}-1$.

The edges of $D$ connect those sequences, that overlap on a fragment of length $z-1$, i.e. when they could be consecutive in $\cT$.
This construction is well known from the so called shift register graph.

\begin{definition}
For a general graph $F$ and weights $w: E_F\to [1,z]$ we define a directed graph $D$ such that 
\begin{itemize}
\item the vertices of $V_D$ are all $z$-tuples $(T_1,\dots,T_z)$ of subsets of $V_F$ such that for all $(t,r)\in E_F: ( t\in T_i \land r\in T_j ) \Rightarrow |i-j|\ge w(t,r)$
\item $((T_1,\dots,T_z),(T_1',\dots,T_z'))\in E_D \Leftrightarrow T_i'=T_{i+1}$ for all $i\in[1,z-1]$.
\end{itemize}
\end{definition}

As the first condition of the above definition mimics (iii) of Lemma~\ref{lem:seq} with $F=T(G)$, 
any sequence $\cT$ that justifies a solution for $(T(G),w,\lambda)$, 
can be transformed into a walk of length $\lambda-z+1$ in $D$. 

In the opposite direction, namely in order to construct a walk in $D$, that corresponds to a valid channel assignment, we need to guarantee also an analogue of the condition (ii) of Lemma~\ref{lem:seq}. 
In other words, each type should occur sufficiently many times in the resulting walk.
Indeed, the construction of $D$ is independent on the function $s$, 
which specifies how many vertices of each type are present in $G$.

In this concern we consider only special walks that allow us to count the occurrences of sets within $z$-tuples.
Observe that $V_D$ contains also the $z$-tuple $\emptyset^z=(\emptyset,\dots,\emptyset)$. In addition, any 
walk of length $\lambda-z+1$ can be converted into a closed walk from $\emptyset^z$ of length $\lambda+z+1$, 
since the corresponding sequence $\cT$ can be padded with additional $z$ empty sets at the front, and another $\emptyset^z$ at the end. From our reasoning, the following claim is immediate:

\begin{lemma}\label{lem:walk}
A closed walk $\cW=W_1,\dots,W_{\lambda+z+1}$ on $D$ where $W_1=W_{\lambda+z+1}=\emptyset^z$,
yields a solution of the {\sc Channel Assignment} problem 
on a $\nd$-uniform instance $G,w,\lambda$ with reflexive $T(G)$, 
if and only if for each $t\in V_{T(G)}$ holds that  $s(t)=|\{W_i: t\in (W_i)_1 \}|$.
\end{lemma}

We found interesting that our representation of the solution resembles the \NP-hardness reduction found by Griggs and Yeh~\cite{l:GY92} (it was briefly outlined in Section~\ref{ss:our}) and later generalized by Bodlaender et al.~\cite{l:BKTL00}. 
The key difference is that in their reduction, a Hamilton path is represented by a sequence of vertices of the constructed graph. 
In contrast, we consider walks in the type graph, which is assumed to be of limited size.

\section{The algorithm}

In this section we prove the following statement, which directly implies our main result, Theorem~\ref{thm:cha}:

\begin{proposition}\label{prop:cha}
Let $G, w$ be a weighted graph, whose weights are uniform with respect to a 
neighborhood diversity partition with $\tau$ classes. 

Then the {\sc Channel Assignment} problem can be decided on $G,w$ and any $\lambda$ in time $2^{2^{O(\tau w_{\max})}} \log n$,
where $n$ is the number of vertices of $G$, provided that $G,w$ are described by the weighted type graph $T(G)$ on $\tau$ nodes.

A suitable labeling of $G$ can be found in additional $2^{2^{O(\tau w_{\max})}} n$ time.
\end{proposition}

\begin{proof}
According to Lemma~\ref{lem:walk}, 
it suffices to find a closed walk $\cW$ (if it exists) corresponding to the desired labeling $l$. 
From the well know Euler's theorem follows that any directed closed walk $\cW$ yields a multiset of edges in $D$ 
that induces a connected subgraph and that satisfies Kirchhoff's law.
In addition, any such suitable multiset of edges can be converted into a 
closed walk, though the result need not to be unique.


For this purpose we introduce an integer variable $\alpha_{(W,U)}$ for every directed edge $(W,U)\in E_D$.
The value of the variable $\alpha_{(W,U)}$ is the number of occurrences of $(W,U)$ in the multiset of edges.

Kirchhoff's law is straightforwardly expressed as:

$$  
\forall W\in V_D:  \sum_{U:(W,U)\in E_D} \alpha_{(W,U)} - \sum_{U:(U,W)\in E_D} \alpha_{(U,W)} = 0
$$

In order to guarantee the connectivity, observe first that an edge $(W,U)$ and $\emptyset^z$ would be 
in distinct components of a subgraph of $D$,
if the subgraph is formed by removing edges that include a cut $C$ between $(W,U)$ and $\emptyset^z$.
Now, the chosen multiset of edges is disconnected from $\emptyset^z$, if for some $(W,U)$ and $C$ holds that $\alpha_{(W,U)}$
has a positive value, while all variables corresponding to elements of $C$ are zeros.
As all variable values are upperbounded by $\lambda$, we express that $C$ is not a cutset for the chosen multiset of 
edges by the following condition:

$$  
\alpha_{(W,U)}-\lambda \sum_{e\in C} \alpha_{e} \le 0
$$

To guarantee the overall connectivity, we apply the above condition for every edge $(W,U)\in E_D$, 
where $W,U\ne \emptyset^z$, and for each set of edges $C$ 
that separates $\{W,U\}$ from $\emptyset^z$. 

The necessary condition expressed in Lemma~\ref{lem:walk} can be stated in terms of variables $\alpha_{(W,U)}$ as

$$
\forall t\in V_{T(G)}:  \sum_{W: t\in (W)_1} \sum_{U:(W,U)\in E_D} \alpha_{(W,U)} = s(t)
$$

Finally, the size of the multiset is the length of the walk, i.e.

$$
\sum_{(W,U)\in E_D} \alpha_{(W,U)} =\lambda+z+1
$$

Observe that these conditions for all $(W,U)$ and all suitable $C$ indeed imply that the $\emptyset^z$ 
belongs to the subgraph induced by edges with positively evaluated variables $\alpha_{(W,U)}$.

Algorithm~\ref{alg:cha} summarizes our deductions.

\begin{algorithm}[h]
\KwIn{A reflexive type graph $T(G)$ whose edges are labeled by $w$ and span $\lambda$.}
\KwOut{A channel assignment $l: G\to [0,\lambda]$ respecting constraints $w$, if it exists.}
\Begin{
\nl Compute $z:=w_{\max}-1$\;
\nl Construct the directed graph $D$\;
\nl Solve the following ILP in variables $\alpha_{(W,U)}: (W,U)\in E_D$:\\
\Indp\Indp
\nl for each $(W,U)\in E_D$:\\
\Indp\Indp
$\alpha_{(W,U)}\ge 0$ \\
\Indm\Indm
\nl for each $W\in V_D$:\\
\Indp\Indp
$\displaystyle\sum_{U:(W,U)\in E_D} \alpha_{(W,U)} - \sum_{U:(U,W)\in E_D} \alpha_{(U,W)} = 0 $\\
\Indm\Indm
\nl for each $(W,U)\in E_D$ and each cutset $C$ between $(W,U)$ and $\emptyset^z$ in $D$:\\
\Indp\Indp
$\displaystyle\alpha_{(W,U)}-\lambda \sum_{e\in C} \alpha_{e} \le 0$ \\
\Indm\Indm
\nl for each $t\in V_{T(G)}$:\\
\Indp\Indp
$\displaystyle\sum_{W: t\in (W)_1} \sum_{U:(W,U)\in E_D} \alpha_{(W,U)} = s(t)$ \\
\Indm\Indm
\nl $\displaystyle\sum_{(W,U)\in E_D} \alpha_{(W,U)} =\lambda+z+1$\;
\Indm\Indm
\nl \eIf {the ILP has a solution}{
\nl find a walk $\cW$ that traverses each edge $(W,U)$ exactly $\alpha_{(W,U)}$ times\;
\nl convert the walk $\cW$ into a labeling $l$ and \Return {$l$}\;
}{\nl \Return {"No channel assignment $l$ of span $\lambda$ exists."}}
}
\caption{Solving the {\sc Channel Assignment} problem.}\label{alg:cha}
\end{algorithm}

To complete the proof, we argue about the time complexity as follows:
\begin{itemize}
\item Line 1 needs $O(|E_{T(G)}|)=O(\tau^2)$ time.
\item As $D$ has $2^{\tau z}$ nodes and $2^{\tau (z+1)}$ edges, line 2 needs $2^{O(\tau z)}$ time.
\item Similarly, conditions at lines 4 and 5
require $2^{O(\tau z)}$ time and space to be composed. Analogously, conditions at lines 7 and 8 involve coefficients that are proportional to the size of the original graph $G$ (namely $\lambda$ and $s(t)$), hence $2^{O(\tau z)}\log n$ time and space is needed here.
\item For line 6, we examine each subset of $E_D$, whether it is a suitable cutset $C$.
There are at most $2^{2^{\tau (z+1)}}$ choices for $C$, so the overall time and space complexity 
for the composition of conditions at line 6 is $2^{2^{O(\tau z)}}\log n$.
\item Frank and Tardos~\cite{m:FrankTardos87} (improving the former result due to Lenstra~\cite{m:Lenstra83}) showed that the time needed to solve the system of inequalities with $p$ integer variables is $O(p^{2.5p+o(p)}L)$, where $L$ is the number of bits needed to encode the input. As we have $2^{O(\tau z)}$ variables and the conditions are encoded in space $2^{2^{O(\tau z)}}\log n$, the time needed to resolve the system of inequalities is $2^{2^{O(\tau z)}}\log n$.
\item A solution of the ILP can be converted into the walk in time $2^{2^{O(\tau z)}}n$, and the same bound applies to the conversion of a walk to the labeling at lines 10 and 11.
\end{itemize}

Observe that if only the existence of the labeling should be decided, the lines 10 and 11 need not to be executed, only an affirmative answer needs to be returned instead.
\end{proof}

We are aware the the double exponential dependency on $\nd$ and $w_{\max}$ makes our algorithm interesting 
mostly from the theoretical perspective. Naturally, one may ask, whether the exponential tower height might be reduced
or whether some nontrivial lower bounds on the computational complexity could be established 
(under usual assumptions on classes in the complexity hierarchy). 

\section{Bounded vertex-cover}\label{s:metaalg}

We utilize the results of the previous sections to derive an \FPT\ algorithm proposed as Theorem~\ref{thm:chaVC}. 

\begin{proof}[Proof of Theorem~\ref{thm:chaVC}]
Given a graph $G$ and its optimal vertex cover $U$, we construct a partition of the vertices of $G$ as follows. 
Next let $I = V(G)\setminus U$ be the independent set of $G$. 
We form a partition of $I$ as follows: For every subset $X\subseteq U$ we define:
$$I_X := \{v\in I\colon \{v,x\}\in E(G) \text{ for } x\in X \text{ and } \{v,x\}\notin E(G) \text{ for } x\notin X\}.$$

Observe that for any $u,v\in I_X$ it holds that $N(u)=X=N(v)$, and hence also $u\sim v$. 
In particular, the optimal neighborhood diversity decomposition of $G$ consists of all nonempty sets $I_X$ together with a suitable partition of $U$. Consequently, $\nd(G)\le 2^{\vc(G)} + \vc(G)$.

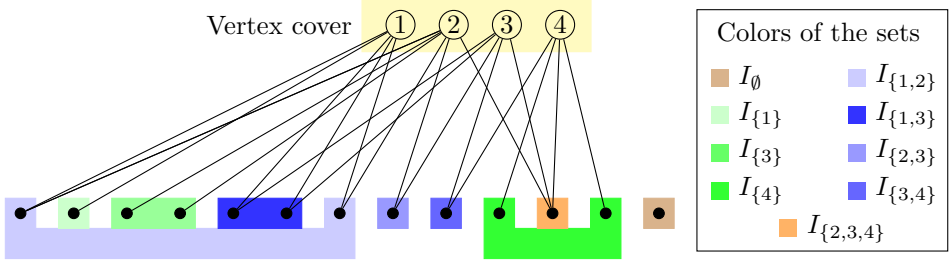
\begin{figure}[ht!]
  \begin{tikzpicture}
  [align=center, node distance=0.7cm]
  \tikzstyle{vrchol}=[circle, draw, fill=black, inner sep=0pt, minimum width=4pt]
  \tikzstyle{vc}=[circle, draw, inner sep=1pt, minimum width=4pt]
  \tikzstyle{legenda}=[rectangle, minimum width=5pt]
  
  \node[fill=yellow!40,thick,draw,yellow!30,
      minimum height=0.7cm,
      minimum width=3cm,
      label=west:Vertex cover 
  ] at (6,2.5) {};
  \begin{scope}
    \node[vc](100) at (5, 2.5) {1};
    \node[vc, right of=100](101) {2};
    \node[vc, right of=101](102) {3};
    \node[vc, right of=102](103) {4};
  \end{scope}

  \draw[fill=blue!20,blue!20] (-0.2,-0.6) -- (-0.2,0.2) -- (0.2,0.2) -- (0.2,-0.2) -- (4,-0.2) -- (4, 0.2) -- (4.4,0.2) -- (4.4,-0.6) -- cycle;
  \draw[fill=blue!80, blue!80] (2.6,-0.2) -- (2.6,0.2) -- (3.7,0.2) -- (3.7,-0.2) --  cycle;
  \draw[fill=green!20, green!20] (1-0.5,-0.2) -- (1-0.5,0.2) -- (.9,0.2) -- (.9,-0.2) --  cycle;
  \draw[fill=green!40, green!40] (1.2,-0.2) -- (1.2,0.2) -- (2.3,0.2) -- (2.3,-0.2) --  cycle;
  \draw[fill=green!80, green!80] (6.1,-0.6) -- (6.1,0.2) -- (6.5,0.2) -- (6.5,-0.2) -- (7.5,-0.2) -- (7.5,0.2) -- (7.9,0.2) -- (7.9,-0.6) -- cycle;
  \draw[fill=blue!40, blue!40] (4.7,-0.2) -- (4.7,0.2) -- (5.1,0.2) -- (5.1,-0.2) --  cycle;
  \draw[fill=blue!60, blue!60] (5.4,-0.2) -- (5.4,0.2) -- (5.8,0.2) -- (5.8,-0.2) --  cycle;
  \draw[fill=orange!60, orange!60] (6.8,-0.2) -- (6.8,0.2) -- (7.2,0.2) -- (7.2,-0.2) --  cycle;
  \draw[fill=brown!60, brown!60] (8.2,-0.2) -- (8.2,0.2) -- (8.6,0.2) -- (8.6,-0.2) --  cycle;
  
  \node[vrchol](1) at (0,0) {};
  \foreach \i [count=\q] in {2,...,13}
    \node[vrchol, right of=\q](\i){};

  \draw  (100) edge (1);
  \draw  (100) edge (2);
  \draw  (100) edge (5);
  \draw  (100) edge (6);
  \draw  (100) edge (7);

  \draw  (101) edge (1);
  \draw  (101) edge (3);
  \draw  (101) edge (4);
  \draw  (101) edge (7);
  \draw  (101) edge (8);
  \draw  (101) edge (11);
  \draw  (101) edge (1);
  
  \draw  (102) edge (5);
  \draw  (102) edge (6);
  \draw  (102) edge (8);
  \draw  (102) edge (9);
  \draw  (102) edge (11);
  
  \draw  (103) edge (9);
  \draw  (103) edge (10);
  \draw  (103) edge (11);
  \draw  (103) edge (12);

  \draw[draw=black] (8.9,-.5) rectangle (12.2,2.7);
  \node at (10.5,2.4) {Colors of the sets};
  \node[legenda,fill=brown!60,label=east:$I_{\emptyset}$] at (9.2,1.8) {};
  \node[legenda,fill=green!20,label=east:$I_{\{1\}}$] at (9.2,1.3) {};
  \node[legenda,fill=green!60,label=east:$I_{\{3\}}$] at (9.2,.8) {};
  \node[legenda,fill=green!80,label=east:$I_{\{4\}}$] at (9.2,.3) {};
  
  \node[legenda,fill=blue!20,label=east:$I_{\{1,2\}}$] at (11,1.8) {};
  \node[legenda,fill=blue!80,label=east:$I_{\{1,3\}}$] at (11,1.3) {};
  \node[legenda,fill=blue!40,label=east:$I_{\{2,3\}}$] at (11,.8) {};
  \node[legenda,fill=blue!60,label=east:$I_{\{3,4\}}$] at (11,.3) {};
  \node[legenda,fill=orange!60,label=east:$I_{\{2,3,4\}}$] at (10.1,-.2) {};
\end{tikzpicture}
  \caption{An example of a neighborhood diversity decomposition based on vertex cover---sets $I_X$ in the bottom.}
\end{figure}

We further refine sets $I_X$, so that the edge-weights became uniform. 
For a set $X = \{x_1,\dots,x_k\}\subseteq U$ and each $k$-tuple of positive integers $\bw = (w_1,\dots,w_k)$ with $0<w_i\le w_{\max}$ for every $1\le i\le k$ we define the set $I_X^\bw$ as
$$I_X^\bw := \{v\in I_X\colon w(v,x_i) = w_i\text{ for } 1\le i\le k\}.$$

Observe that any refinement of a neighborhood diversity decomposition is again a decomposition. We now estimate the number of types of the refined decomposition. The number of types of the refined decomposition can be upper-bounded by $\vc(G) + (2^{\vc(G)})w_{\max}^{\vc(G)}.$ To finish the proof we apply Proposition~\ref{prop:cha} on the refined decomposition.
\end{proof}

\section{NLC-uniform channel assignment}

One may ask whether the concept of $\nd$-uniform weights could be extended to broader graph classes.
We show, that already its direct extension to graphs of bounded clique width makes the 
{\sc Channel Assignment} problem \NP-complete. Instead of clique width we express our 
results in terms of NLC-width~\cite{t:Wanke94} (NLC stands for node label controlled). 
The parameter NLC-width is linearly dependent on clique width, but it is technically simpler.

We now briefly review the related terminology. 
A NLC-decomposition of a graph $G$ is a rooted tree whose leaves are in one-to-one correspondence 
with the vertices of $G$. For the purpose of inserting edges, each vertex is given a label 
(the labels for channel assignment are now irrelevant), which may change during the construction of the graph $G$.
Internal nodes of the tree are of two kinds: \emph{relabel} nodes and \emph{join} nodes.

Each relabel node has a single child and as a parameter takes a mapping $\rho$ on the set of labels.
The graph corresponding to a relabel node is isomorphic to the graph corresponding to its child,
only $\rho$ is applied on each vertex label.

Each join node has a two children and as a parameter takes a symmetric binary relation $S$ on the set of labels.
The graph corresponding to a relabel node is isomorphic to the disjoint union of the two graphs $G_1$ and $G_2$
corresponding to its children, where further edges are inserted as follows: 
$u\in V_{G_1}$ labeled by $i$ is made adjacent to $v\in V_{G_2}$ labeled by $j$ if and only if $(i,j)\in S$.

The minimum number of labels needed to construct at least one labeling of $G$ 
in this way is the NLC width of $G$, denoted by $\nlc(G)$.

Observe that $\nlc(G)\le \nd(G)$ as the vertex types could be used as labels for the corresponding vertices
and the adjacency relation in the type graph could be used for $S$ in all join nodes. 
In particular, in this construction the order of performing joins is irrelevant and no relabel nodes are needed.

\begin{definition}
The edge weights $w$ on a graph $G$ are \emph{$\nlc$-uniform} w.r.t. a particular 
NLC-decomposition, if $w(u,v)=w(u', v')$ 
whenever edges $(u,v)$ and $(u',v')$ are inserted during the same join operation
and at he moment of insertion $u,u'$ have the same label in $G_1$ and $v,v'$ have the same label in $G_2$.
\end{definition}

Observe that our comment before the last definition justifies 
that weights that are uniform w.r.t. a neighborhood diversity decomposition 
are uniform also w.r.t. the corresponding NLC-decomposition.

Gurski and Wanke showed that NLC-width remains bounded when taking powers~\cite{t:GW09}.
It is well known that NLC-width of a tree is at most three. Fiala et al. proved
that {\sc $L(3,2)$-labeling} is \NP-complete on trees~\cite{l:FGK08}. To combine these facts together 
we show that the weights on the graph arising from a reduction of the $L(3,2)$-labeling 
on a tree to {\sc Channel Assignment} are $\nlc$-uniform.

\begin{theorem}
The {\sc Channel Assignment} problem is \NP-complete on graphs with edge 
weights that are $\nlc$-uniform w.r.t. a NLC-decomposition of width at most four.
\end{theorem}

\begin{proof}
Let a tree $T$ be an instance of the {\sc $L(3,2)$-labeling} problem.

By induction on the size of $T$ we show that $T^2$ allows an NLC-decomposition  
such that the weights $w$ prescribed by the reduction of {\sc $L(3,2)$-labeling}
to {\sc Channel Assignment} are $\nlc$-uniform.

Assume for the induction hypothesis that such NLC-decomposition 
exists for every tree $T'$ on less than $n$ vertices, 
where the labels of $T'$ are distributed as follows:
assume that $T'$ is rooted in a vertex $r'$, then $r'$ is labeled by $1$, its direct neighbors by $2$ and all other vertices by $3$.

Such decomposition clearly exists for a tree on a single vertex.

Now consider a tree $T$ on $n\ge 2$ nodes. Choose and edge $(r',r'')$ arbitrarily and define two trees $T'$ and $T''$ as the components of $T\setminus (r',r'')$, where $T'$ contains $r'$ and vice versa.

\begin{figure}[ht!]
  \includegraphics{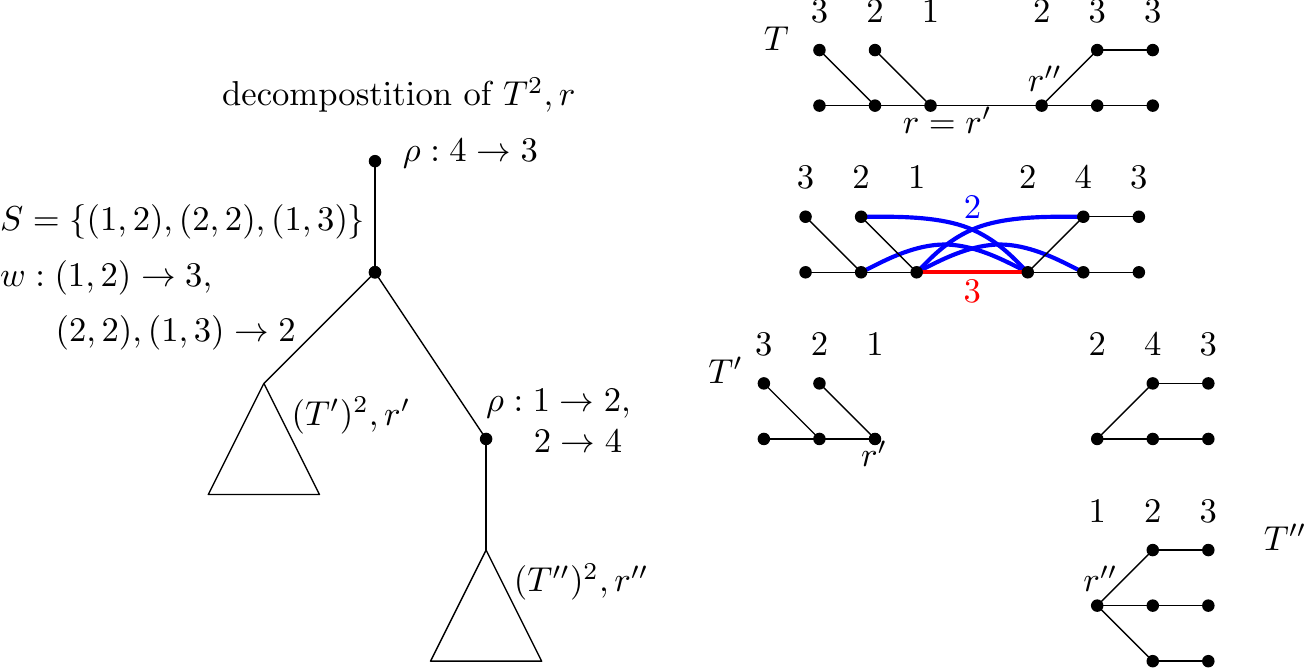}
  \caption{Recursive step in the construction of NLC-decomposition. The original edges of $T,T'$ and $T''$ are in black, 
	         only the added weighted edges of $T^2$ are in color.}
	\label{fig:nlc}
\end{figure}

By induction hypothesis $T'$ and $T''$ allow NLC-decompositions of the desired properties.
Before we join trees $T'$ and $T''$ together, we change labels in $T''$ as $1\to 2, 2\to 4$.
At the join will insert the following weighted edges: of weight 3 between vertices labeled 1 in $T'$ and 2 in $T''$, 
and of weight 2 between vertices of labels 2 and 2, and between 1 and 4, respectively.
Finally, we relabel $4\to 3$ and promote $r'$ to be the root $r$ of $T$. 
All steps are depicted in Fig.~\ref{fig:nlc}.  Observe that the result of this construction is $T^2$
with appropriate $\nlc$-uniform weights, and that its labeling satisfies all conditions of the induction hypothesis.
\end{proof}

\section{Conclusion}\label{s:con}

We have shown an algorithm for the {\sc Channel Assignment} problem on $\nd$-uniform instances 
and several complexity consequences for
{\sc $L(p_1,\dots p_k)$-labeling} problem. In particular, Theorem~\ref{thm:Lp} extends known results
for $L(p,1)$-labelings problem to labelings with arbitrarily many distance constraints, answering
an open question of~\cite{l:FGK09}. Simultaneously, we broaden the considered graph classes by restricting
neighborhood diversity instead of vertex cover.

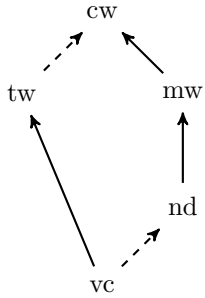
\begin{figure}[!ht]
  \begin{minipage}[c]{0.3\textwidth}
    \begin{tikzpicture}
  [align=center, node distance=1.5cm]
  \tikzstyle{parameter}=[inner sep=5pt]

  \node[parameter](cw) {cw};
  \node[parameter, below left of=cw](tw) {tw};
  \node[parameter, below right of=cw](mw) {mw};
  \node[parameter, below of=mw](nd) {nd};
  \node[parameter, below left of=nd](vc) {vc};

  \draw[->,line width=5mm,thick,dashed,>=stealth'] (tw) -- (cw);
  \draw[->,thick,dashed,>=stealth'] (vc) -- (nd);
  \draw[->,thick,>=stealth'] (mw) -- (cw);
  \draw[->,thick,>=stealth'] (nd) -- (mw);
  \draw[->,thick,>=stealth'] (vc) -- (tw);
  
\end{tikzpicture}
  \end{minipage}\hfill
  \begin{minipage}[c]{0.65\textwidth}
    \caption{A map of assumed parameters. Full arrow stands for linear upper bounds, while dashed arrow stands for exponential upper bounds.}
    \label{fig:parameterMap}
  \end{minipage}
\end{figure}

While the main technical tools of our alogorithms are bounded-dimension ILP programs, ubiquitous in the \FPT{} area,
the paper shows an interesting insight on the nature of the labelings over the type graph and the necessary
patterns of such labelings of very high span. Note that the span of a graph is generally not bounded by any of the
considered parameters and may be even proportional to the order of the graph.

Solving a generalized problem on graphs of bounded neighborhood diversity is a viable method
for designing \FPT{} algorithms for a given problem on graphs of bounded vertex cover, as demostrated by this and previous papers.
This promotes neighborhood diversity as a parameter that naturally generalizes the widely studied parameter vertex cover.

We would like to point out that the parameter {\em modular width}, proposed by Gajarský, Lampis and Ordyniak~\cite{t:GLO13}, offers further generalization of neighborhood diversity towards the clique width~\cite{t:CO00} (dependencies between these graph parameters are depicted in Fig.~\ref{fig:parameterMap}).

As an interesting open problem we ask whether it is possible to strengthen our results to graphs of bounded modular width
or whether the problem might be already \NP-complete for fixed modular width, as is the case with clique width.
For example, the {\sc Graph Coloring} problem ILP based algorithm for bounded neighborhood diversity translates
naturally to an algorithm for bounded modular width. On the other hand, there is no apparent way how our labeling
results could be adapted to modular width in a similar way.


\bibliography{main}

\end{document}